\documentclass{article}

\usepackage{arxiv}
\usepackage{amsmath}
\usepackage{amsthm}
\usepackage[utf8]{inputenc} 
\usepackage[T1]{fontenc}    
\usepackage{hyperref}       
\usepackage{url}            
\usepackage{booktabs}       
\usepackage{amsfonts}       
\usepackage{nicefrac}       
\usepackage{microtype}      
\usepackage{cleveref}       
\usepackage{lipsum}         
\usepackage{graphicx}
\usepackage{natbib}
\usepackage{doi}
\newtheorem{lemma}{Lemma}

\title{Greedy SLIM: A SLIM-Based Approach For Preference Elicitation}

\newif\ifuniqueAffiliation
\uniqueAffiliationtrue

\ifuniqueAffiliation 
\author{ Claudius Proissl\\
	University of Stuttgart\\
	Germany \\
	\And
	Amel Vatic \\
	University of Stuttgart\\
	Germany\\
	\And
	Helmut Waldschmidt \\
	University of Stuttgart\\
	Germany\\
}
\else
\usepackage{authblk}

\newbox{\orcid}\sbox{\orcid}{\includegraphics[scale=0.06]{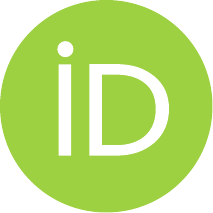}} 
\author[1]{%
	\href{https://orcid.org/0000-0000-0000-0000}{\usebox{\orcid}\hspace{1mm}David S.~Hippocampus\thanks{\texttt{hippo@cs.cranberry-lemon.edu}}}%
}
\author[1,2]{%
	\href{https://orcid.org/0000-0000-0000-0000}{\usebox{\orcid}\hspace{1mm}Elias D.~Striatum\thanks{\texttt{stariate@ee.mount-sheikh.edu}}}%
}
\affil[1]{Department of Computer Science, Cranberry-Lemon University, Pittsburgh, PA 15213}
\affil[2]{Department of Electrical Engineering, Mount-Sheikh University, Santa Narimana, Levand}
\fi

\renewcommand{\undertitle}{~´}
\renewcommand{\shorttitle}{Greedy SLIM}

\hypersetup{
pdftitle={A template for the arxiv style},
pdfsubject={q-bio.NC, q-bio.QM},
pdfauthor={David S.~Hippocampus, Elias D.~Striatum},
pdfkeywords={First keyword, Second keyword, More},
}

\begin{document}
\maketitle

\begin{abstract}
	Preference elicitation is an active learning approach to tackle the cold-start problem of recommender systems. Roughly speaking, new users are asked to rate some carefully selected items in order to compute appropriate recommendations for them. 
	
	To the best of our knowledge, we are the first to propose a method for preference elicitation that is based on SLIM \cite{ning2011slim}, a state-of-the-art technique for top-N recommendation.
	
	Our approach mainly consists of a new training technique for SLIM, which we call Greedy SLIM. This technique iteratively selects items for the training in order to minimize the SLIM loss greedily.
	
	We conduct offline experiments as well as a user study to assess the performance of this new method. The results are remarkable, especially with respect to the user study. We conclude that Greedy SLIM seems to be more suitable for preference elicitation than widely used methods based on latent factor models.
\end{abstract}


\section{Introduction}
Providing accurate top-N recommendations is an important task for many online services. Among the most successful approaches to this problem are those in the field of collaborative filtering, which process user-item interactions in order to compute recommendations. The accuracy of these approaches is known to suffers for users with only a few item interactions, which is always the case if new users enter the system. 

This issue is known as the user cold-start problem \cite{lam2008addressing} and it can be addressed in many different ways \cite{ElahiMehdi2019UPER}. Perhaps the most popular approach is to augment the input parameters with other sources of information by, for instance, combining collaborative filtering with content-based filtering \cite{soboroff1999combining,NIPS2017_51e6d6e6,LIKA20142065} or by incorporating cross-domain knowledge about the new users \cite{fernandez2016alleviating,10.1145/2645710.2645777}. 

A different way to tackle the cold-start problem is to ask the new users to rate some carefully selected items during an onboarding process. We call this the questionnaire approach, also known as preference elicitation \cite{parapar2021diverse,sepliarskaia2018preference,ElahiMehdi2019UPER}. This method aims to create a first basis of ratings for each user with which the recommender system can produce meaningful results. 

In this work we address the questionnaire approach without claiming that it is the preferred way. Our humble advice to practitioners is to use any information about the users that is available. There are, however, many conceivable scenarios where this information is simply insufficient to work with and then the questionnaire approach may be the best option. Another strong point of this way is that it inherits all the benefits of collaborative filtering: it can be widely applied, no domain-specific knowledge is necessary, and the users can stay anonymous. 

The general question we address in this work is, given a top-N recommender of your choice, which items should new users be asked to rate in order to maximize the accuracy of the subsequently computed recommendations? This is a classic problem in the field of active learning and similar questions have been raised in many other publications before \cite{towards, adaptive, sepliarskaia2018preference, rokach2012initial,fonarev2016efficient,Wang_Wu_Wang_2017}. Clearly, the answer may depend on the chosen recommender system. In this aspect, the existing literature is surprisingly confident that latent factor models (LFM, also known as matrix factorization; see, for instance, \cite{Koren2022}) are the correct choice.

In this work we propose a different approach and use SLIM \cite{ning2011slim} as recommender system. SLIM is a well-established method for top-N recommendation in the field of collaborative filtering. Ever since its publication in 2011, SLIM has attracted a lot of attention in the scientific community and has been inspiration for similar approaches \cite{kabbur2013fism,10.1145/2645710.2645756,steck2019embarrassingly,steck2019embarrassingly,steck2020admm}. 

Still, one may wonder why we select a method that is more than ten years old. First, to us, SLIM is more a concept than a concrete method, just as the LFM is and there are modern and very efficient incarnations of it \cite{steck2019embarrassingly,steck2020admm}. Second, there are available implementations of SLIM everywhere. It is therefore safe to bet that SLIM will stay important for the next couple of years.

The concrete question we address in this work is how SLIM can be used to construct efficient questionnaires. The answer we propose in this work is that the typical way to train SLIM is not suitable to tackle the cold-start problem (Section \ref{sec:firstapproach}). We therefore discuss in Section \ref{sec:greedyslim} a new preprocessing routine that trains SLIM iteratively item by item, each turn greedily selecting the item that minimizes the loss. We call this approach Greedy SLIM. Our training, thus, already orders the items by importance and our questionnaire simply consists of the $k$ most important items (where $k$ can be chosen arbitrarily).

In Section \ref{sec:experiments} we compare our approach with a popular existing preference elicitation method and in Section \ref{sec:userstudy} we present our results from a small user study we conducted.

\section{Related Work}
\label{sec:relatedwork}
We are not the first to adapt the training phase of SLIM to a specific problem setting. For instance, in \cite{steck2020admm} a training method based on the Alternating Direction Method of Multipliers \cite{boyd2011distributed} is proposed. This approach decouples the run time of the training phase from the number of users, which is especially important if the number of users is much larger than the number of items. 

Improving the run time of the training phase is also the main objective of the approach in \cite{steck2019embarrassingly}, which slightly changes the definition of SLIM in order to apply more efficient training methods.

In principle, we think that the ideas presented in this work can be combined with any autoencoder-based top-N recommender. We chose SLIM as we think that it is the most established approach.

Preference elicitation is a very active research direction and in our humble opinion it is hard to tell which methods are state-of-the-art. Most existing approaches are based on LFMs. Some examples are \cite{towards, sepliarskaia2018preference, rokach2012initial,fonarev2016efficient,Wang_Wu_Wang_2017}. For instance, in \cite{towards} a probabilistic method is proposed that tries to find the optimal trade-off between exploration and exploitation when selecting the next question. This strategy is motivated by the multi-armed bandit problem. The general idea is to assign weights to all existing users denoting the probability that the new user rates items similarly. The parameters for the new user are then obtained by taking the weighted sum of the existing users' parameters. We call this approach in the following $Q_{\text{Bandit}}$ for bandit questionnaire.

The authors of \cite{adaptive} present a similar method based on decision trees. Their algorithm chooses the questions greedily by taking the question that minimizes the expected root mean squared error (RMSE) of the new user's predicted ratings after receiving the answer to the selected question. The shortcoming of this approach is that the RMSE is not a good measure for the performance of top-N recommender systems as it was shown in \cite{cremonesi2010performance}. It is, however, one of the few existing techniques that is not based on the LFM.

Another interesting LFM-based approach is described in \cite{sepliarskaia2018preference}. Very roughly speaking, it selects questions based on a kind of binary search in order to narrow down the region in the latent factor space where the feature vector of the new user may live. We found that this approach was very appealing from a theoretical point of view. 

Despite these and other interesting techniques we found it hard to find appropriate baselines for our approach. First of all, while there is a lot of literature about the cold-start problem in general, the questionnaire approach is less popular. Furthermore, the reproducibility issue within the recommender systems community is also an issue when it comes to questionnaires. We had to implement all the baselines we wanted to consider on our own. For \cite{adaptive} and \cite{sepliarskaia2018preference} we, unfortunately, obtained poor results. We therefore decided to not include them in our official experiments as there is always the possibility of a misinterpretation or an implementation error.

We chose $Q_{\text{Bandit}}$ as the main baseline for our experiments in Section \ref{sec:experiments} and \ref{sec:userstudy}. We found that this method is popular and a good representative of many other LFM-based approaches. Furthermore, the method is intuitive and, hence, we are confident that our implementation is correct.

\section{Preliminaries}
\label{sec:preliminaries}
In the remainder of this work we assume that we are given a set of $m$ users $\mathcal{U} :=\lbrace u_1,\,u_2,\dots,\,u_m\rbrace$ and a set of $n$ items $\mathcal{I}:=\lbrace i_1,\,i_2,\dots,\,i_n\rbrace$. Furthermore, we have a user-item interaction matrix $X$ of size $m\times n$. The definition of a user-item interaction differs from application to application. In our case user $u$ interacts with item $i$ by setting the rating $r_{ui}\in[1,\,5]$. Therefore, the values $x_{ui}$ of $X$ are either $0$ if user $u$ has not rated item $i$ yet or $x_{ui} = r_{ui}$.

Our vectors are column vectors and written with bold lowercase letters. For any matrix $A$ we write $\textbf{a}_j$ for the $j$-th column of $A$ and $\textbf{a}^T_j$ for its $j$-th row. We write $A^T$ for the transposed matrix of $A$. Furthermore, we write $\lVert\cdot\rVert_1$ for the $L_1$-norm and $\lVert\cdot\rVert_F$ for the Frobenius norm. For any set $S$ we write $\vert S\vert$ to denote its size.
\subsection{Questionnaire}
We define a questionnaire to be a function $Q:~\mathbb{R}^{m\times n} \times \mathcal{U}\rightarrow \mathcal{I}$ that takes as input the user-item interaction matrix $X$ and a user $u$. The output of this function is an item $i$ that the user $u$ should rate. The questionnaire is therefore allowed to reveal some unknown ratings of the input user $u$.

We always consider questionnaires in combination with a recommender system $R:~\mathbb{R}^{m\times n} \times \mathcal{U} \times N\in\mathbb{N}\rightarrow \mathcal{I}^N$. We interpret recommender systems as a function $R$ that takes as input the user-item interaction matrix $X$, a user $u$ and a number $N$ and returns $N$ sorted items as recommendations for user $u$. For a given recommender system, our goal is to find the right questionnaire in order to maximize the accuracy of the recommender system for the input user $u$. In this work we try to find a good questionnaire for SLIM.

An important application of questionnaires is when the input user $u$ is new to the system, which means the row $\textbf{x}^T_u$ of $X$ is a zero-vector. This is also the scenario we consider in this work. In such a case, it makes sense to distinguish between two types of questionnaires, static and dynamic questionnaires. Static questionnaires ask the same questions to all new users whereas the output of dynamic questionnaires depends on the already given answers of the user. 

While dynamic questionnaires are more flexible with respect to the asked questions, the computation of the questions often needs to be conducted online. This means that the questions must be computed fast. Static questionnaires can be computed in a preprocessing phase such that much more complex calculations are affordable. Our approach describes a static questionnaire but we believe that our ideas can be easily extended to the dynamic setting.

In this work we assume that a question always consists of a single item that the user should rate. Clearly, there are other types of questions that are at least as plausible. For instance, a question could also consist of two items and the users should tell which one they prefer. In \cite{towards} a way is described to extend $Q_{\text{Bandit}}$ to this question type. There are concerns that so called absolute questions, questions about one item, are too difficult to answer accurately. We do see that point as well but prefer to keep this dimension of complexity out for the sake of readability and compactness.

\subsection{SLIM}
\label{subsec:slim}
In SLIM \cite{ning2011slim} we are looking for a non-negative matrix $W$ of size $n\times n$ such that $XW \approx X$. As such, the problem could be trivially solved by choosing $W$ to be the identity matrix. We therefore additionally require the diagonal of $W$ to be zero. A suitable matrix $W$ reveals interaction similarities among items. For instance, if items $i_1$ and $i_2$ had identical user interactions, we could set $w_{i_1i_2} = w_{i_2i_1} = 1$. The non-negativity constraint ensures that $W$ only reflects positive interaction correlations. In \cite{steck2020admm} it is shown that for some datasets it is beneficial to drop this constraint. We call $W$ the SLIM matrix.
\newpage
The SLIM loss function $l_{SLIM}:~\mathbb{R}^{n\times n}\rightarrow \mathbb{R}$ is defined as
\begin{equation}
\begin{aligned}
l_{SLIM}\left(W\right):=\lVert X - XW\rVert_F^2 + \lambda_F \lVert W\rVert_F^2 + \lambda_1 \lVert W\rVert_1,
\end{aligned}
\end{equation}
where $\lambda_F$ and $\lambda_1$ are fixed parameters. The complete optimization problem is given below.
\begin{equation}
\begin{aligned}
&\min_W &l_{SLIM}\left(W\right)\\
&\text{subject to} &W\ge 0\\
&&\text{diag}\left(W\right) = 0
\end{aligned}
\end{equation}

In \cite{ning2011slim} the authors propose coordinate descent \cite{friedman2010regularization} in order to compute $W$. Another approach based on the Alternating Direction Method of Multipliers is presented in \cite{steck2020admm}. In this work we propose a greedy strategy to compute $W$, which we present in Section \ref{sec:greedyslim}.

In the following we discuss how we can compute top-N recommendations for any user $u\in\mathcal{U}$ using the SLIM matrix $W$. We define the predicted relevance $\tilde{r}_{ui}$ of item $i\in\mathcal{I}$ for user $u$ as
\begin{align*}
	\tilde{r}_{ui} := \textbf{x}_u^T \textbf{w}_i,
\end{align*}
where $\textbf{x}_u^T$ is the row of matrix $X$ that corresponds to the item interactions of user $u$ and $\textbf{w}_i$ is the $i$-th column of matrix $W$. We sort the items $i\in\mathcal{I}$ with $x_{ui} = 0$ by their predicted relevance in descending order and show user $u$ the list of the first $N$ items ($x_{ui} = 0$ because we restrict our recommendations to unknown user-item interactions).

\subsection{Evaluation Method}
\label{subsec:evalmethod}

\subsubsection{Evaluation Metrics}
We mainly use the Normalized Discounted Cumulative Gain (NDCG) in order to compare different top-N recommendations. The NDCG is designed to measure the relevance of search results, which fits very well to our setting. We set the gain $g_{ui}$ of user $u$ for item $i$ to be $2^{x_{ui}} - 1$ in order to emphasize that a high rating is much more important than a medium rating.

Let $\mathcal{I}_R\subseteq \mathcal{I}$ be the item subset from which the recommender system is allowed to draw the recommendations. This may not be the complete item set $\mathcal{I}$ as we sometimes exclude items (discussed in Section \ref{subsubsec:offlineexperiments}). The cumulative gain $CG:~\mathbb{R}^{m\times n}\times\mathcal{U}\times2^{\mathcal{I}_R}$ takes the interaction matrix $X$, a user $u$ and an item subset $I\subseteq \mathcal{I}_R$ as input and is defined as

\begin{align*}
	CG\left(X,\,u,\,I\right):= \sum_{i\in I} 2^{x_{ui}} - 1.
\end{align*}

The discounted cumulative gain $DCG$ turns the subset $I$ into a sorted list and weights the gain depending on the position of the item. This is motivated by the goal to put the best search results on top of the list.

\begin{align*}
	DCG\left(X,\,u,\,I\right):= \sum_{1\le j \le \vert I\vert} \frac{2^{x_{uI[j]}} - 1}{\log_2\left(j+1\right)}.
\end{align*}

The NDCG is the DCG normalized by the maximum DCG among all subsets $I\subseteq \mathcal{I}_R$ (for a fixed user $u$). This ensures that the NDCG is a number between 0 and 1. We write NDCG@$N$ to emphasize that $\vert I\vert$ is $N$. 

Other metrics we consider are Precision@$N$ and Recall@$N$. Let $I_u$ be the set of items rated by user $u$ and let $I_R$ be the recommendation for this user. Precision is then defined as
\begin{align*}
	\frac{\vert I_R\cap I_u\vert}{\vert I_R\vert},
\end{align*}
while Recall is defined as
\begin{align*}
	\frac{\vert I_R\cap I_u \vert}{\vert I_u\vert}.
\end{align*}

\subsubsection{Offline Experiments}
\label{subsubsec:offlineexperiments}
Our offline experiments of this work have the following structure. We first randomly split the users $\mathcal{U}$ into disjoint training and test sets $\mathcal{U}_{\text{train}}$ and $\mathcal{U}_{\text{test}}$ with $\vert\mathcal{U}_{\text{test}}\vert = 0.1\cdot \vert\mathcal{U}\vert$. The corresponding interaction matrices $X_{\text{train}}$ and $X_{\text{test}}$ only contain the interactions of the respective users. 

For each pair $(Q,\,R)$ of questionnaires and recommender systems we want to consider we conduct the following experiment for each test user $u\in \mathcal{U}_{\text{test}}$. We add a zero row for the user $u$ to the interaction matrix $X_{\text{train}}$ and call the thereby obtained matrix $X'$. We then iteratively call the questionnaire $Q\left(X',\,u\right)$, which returns an item $i$, and copy the user-item interaction $r_{ui}$ from $X_{\text{test}}$ to $X'$. When we call $Q\left(X',\,u\right)$ the next time, the row of user $u$ may therefore look different. After doing this (arbitrary) $k$ times, we call our recommender system $R\left(X',\,u,\,N\right)$ to evaluate the accuracy of the top-N recommendation after $k$ questions. 

Let $I_Q$ be the set of items asked by the questionnaire and let $I_R$ be the set of recommended items. We always require $I_Q \cap I_R = \emptyset$ to avoid trivial recommendations.

We found that this procedure simulates the onboarding process of a new user $u$ well and is less biased than other approaches where the user behavior is simulated by some trained model. The following shortcomings remain, however. First, the information contained in $X_{\text{test}}$ may be incomplete because the users do not rate all the items they know. Second, a real new user most likely knows less items than the users contained in $X_{\text{test}}$. We found that these points might have a severe impact on the results, which is why we additionally conducted an online user study that is discussed in Section \ref{sec:userstudy}.

It is common practice to evaluate how well top-N recommender systems perform if they are restricted to lesser known items. Following \cite{cremonesi2010performance}, we refer to the set of most popular items that represent 33\% of the ratings as short-head items. The remaining items are called long-tail items. Accurately recommending short-head items is relatively simple, as we demonstrate in Section \ref{sec:firstapproach}, and often rather pointless as most users, even new users, are likely to be aware of these items. It is much more challenging to find accurate recommender systems for long-tail items, which is why we conduct separate experiments for these. In this setting the questionnaires are still allowed to contain short-head items, of course, but the top-N recommendation is restricted to long-tail items. 

Note that in our experiments we do not apply any kind of item sampling to avoid the problems reported in \cite{krichene2022sampled}. 

\subsubsection{Reproducibility}
Our source code can be found here\footnote{Source code: \url{https://osf.io/myh2q/?view_only=a9a747dee8704203ae594ae0a8cc88f8}}.

\subsection{Bandit Questionnaire}
\label{subsec:qbandit}
As $Q_{\text{Bandit}}$ \cite{towards} is our main baseline algorithm in Section \ref{sec:experiments} and \ref{sec:userstudy}, we would like to discuss it in more detail. The approach assumes that a trained LFM is available without specifying how exactly this model is obtained. An LFM mainly consists of feature vectors, for the users as well as for the items (see, for instance, \cite{Koren2022} for an overview). If we believe that an LFM approximates the rating behavior of the users well, our goal should be to find out the feature vectors of new users in order to predict their ratings.

The idea of $Q_{\text{Bandit}}$ is to approximate the feature vector of the new user by a weighted sum of the feature vectors of the existing users. We can then use this approximation to compute the recommendation for the user in a straightforward manner.

The approach takes a probabilistic point of view and interprets the weighted sum of the feature vectors as the expected feature vector of the new user. This can be done in a mathematically strict manner with the assumption that all users rate according to the LFM plus some random noise that is normally distributed with mean zero. Each existing user is then weighted by the probability that he or she rates as the new user (times a normalization factor).

It remains to specify a strategy, which questions are asked to the new user. The authors of \cite{towards} propose several plausible ways, in the nature similar to the baseline approaches we discuss in Section \ref{sec:firstapproach}. We only consider the method that works best according to \cite{towards}, which they call Thompson Sampling \cite{NIPS2011_e53a0a29}. In this context\footnote{At least, that's our interpretation as in \cite{towards} it is not clearly specified.}, Thompson Sampling means that for each question we draw an existing user from the probability distribution specified by the mentioned weights. Initially, this is the uniform distribution. We then pick the item as question that receives the best rating of the drawn user according to the LFM. 

Note that the way questions are selected would also be a reasonable way to construct recommendations. From an information theoretical point of view this makes sense (to some extend) because the new user most likely only knows a fraction of the items. We should therefore try to find items the user knows to increase the information gain per question. 

\section{A First Approach}
\label{sec:firstapproach}

In this section we would like to demonstrate the difficulty of finding good questionnaires for SLIM. We discuss some straightforward strategies and show their (poor) performance in experiments with our Movielens dataset (ML-25, see Table \ref{tab:data}).

We train SLIM as described in Section \ref{subsec:evalmethod} using the implementation of \cite{slim}, version 2.0, with the default parameters $\lambda_1 = 1$ and $\lambda_F = 1$.\footnote{Unfortunately, we did not have the computational capacity for a proper parameter optimization as the training phase of SLIM is computationally expensive.}

A basic but important approach for a questionnaire is to simply pick the most popular items \cite{rashid2008learning}. Our first questionnaire $Q_{\text{Pop}}$ therefore selects among all items $\mathcal{I}\setminus I_u$ the item with most ratings, where $I_u$ is the set of items the input user $u$ has already rated. We also consider the questionnaire $Q_{\text{Var}}$ that selects items based on a trade-off between high entropy and high popularity \cite{adaptive}. Note that $Q_{\text{Pop}}$ and $Q_{\text{Var}}$ are not at all adapted to SLIM.

We therefore present a more sophisticated approach we call $Q_{\text{Greedy}}$. Let $W$ be the SLIM matrix we obtained from our training with $X_{\text{train}}$ and let $I$ be any subset of the  items $\mathcal{I}$. We define $W_{I}$ to be equal to $W$ for every row $i\in I$ and to be filled with zeros otherwise. We can think of $I$ to be our questionnaire that reveals columns of the interaction matrix, which is equivalent to revealing rows of $W$. 

We start with $I:=\lbrace\rbrace$ and in each step we add the item $i$ to $I$ such that $l_{SLIM}\left(W_I\right)$ is minimized. The motivation behind this approach is that we want to maximize the information gain of every question, which we achieve greedily by minimizing the SLIM loss. 

Note that $Q_{\text{Greedy}}$ has a similar flavor as the approach we present in Section \ref{sec:greedyslim} but is conceptually very different. Here we take a trained matrix $W$ and greedily order the items while the approach in Section \ref{sec:greedyslim} greedily computes the matrix $W$.

To get a better feeling how well the mentioned approaches perform, we add a very simple, static recommender system $R_{\text{Gain}}$ that recommends the same set of items to all users in $\mathcal{U}_{\text{test}}$. It sorts the items by their sum of gain (see Section \ref{subsec:evalmethod}) for the users in $\mathcal{U}_{\text{train}}$ in descending order and returns the first $N$ items as recommendation. We've tried several other static recommender systems and found that this one performs best with respect to NDCG. This approach achieves an average NDCG@10 of $0.314$ for all items and of $0.064$ for long-tail items, which confirms the well-known observation that recommending long-tail items is much more challenging. Clearly, any sophisticated approach must outperform this baseline to be of relevance. 

\begin{figure}
	\includegraphics*[width=1.\columnwidth]{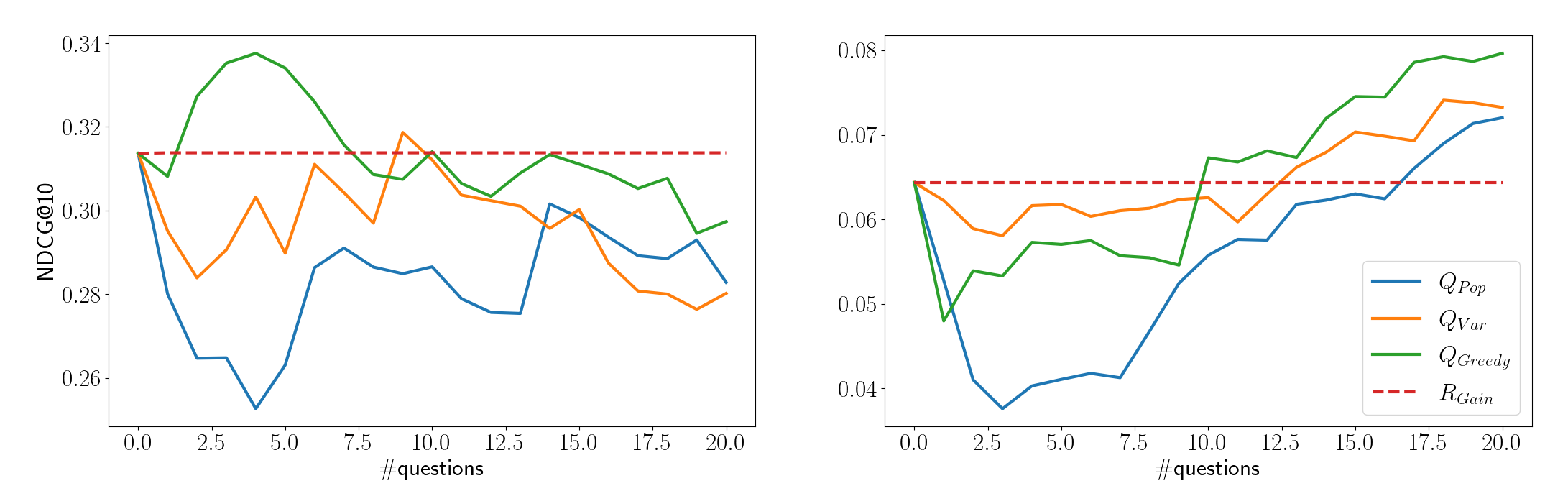}
	\caption{Performance of different SLIM questionnaires and $R_{\text{Gain}}$ for ML-25. Left: all items, right: long-tail items. While for long-tail items the NDCG increases with the number of questions, this is not the case for all items.}
	\label{fig:firstapproach}
\end{figure}

Figure \ref{fig:firstapproach} shows the results of our experiments. Considering all items (left) we can observe that, surprisingly, the recommendation accuracy decreases with the number of questions. The questionnaires mainly consist of popular items that would have been good candidates for the recommendation. As we restrict the recommendation list to items not included in the questionnaire, asking questions about popular items may have a negative effect on the recommendation performance if the recommender system is unable to generalize the gained information. Another interesting observation is that only $Q_{\text{Greedy}}$ is able to achieve better results than $R_{\text{Gain}}$.

Looking at long-tail items (Figure \ref{fig:firstapproach}, right) we can observe that the NDCG first drops below the static approach $R_{\text{Gain}}$, which is used as fallback strategy for all other approaches. However, after the first five questions the performance of the SLIM questionnaires improves with each question as expected. Again, $Q_{\text{Greedy}}$ seems to be the best approach by beating $R_{\text{Gain}}$ after $10$ questions. 

In summary, even though we are partly able to obtain better results than $R_{\text{Gain}}$, we do not think that the presented approaches are practical. The reason is that it takes 15 to 20 questions until we obtain a clear performance benefit, which we consider to be too slow in practice.

\section{Greedy SLIM}
\label{sec:greedyslim}
Motivated by the mediocre performance of SLIM-based questionnaires (see Section \ref{sec:firstapproach}) we present a new approach for the training phase of SLIM that is much more suitable for the questionnaire setting. Our method constructs the SLIM matrix $W$ row by row, each time selecting the item $i\in\mathcal{I}$ that minimizes the SLIM loss. We therefore propose a greedy algorithm for the SLIM training, which is why we call it Greedy SLIM, or GSLIM. Our approach tries to maximize the information gain with every additional row of $W$, which is exactly what a questionnaire tries to achieve. After the training phase our questionnaire $Q_{\text{GSLIM}}$ orders the items used as questions in the same way the rows were added to $W$. After receiving the answers we can use matrix $W$ to compute the recommendations as explained in Section \ref{subsec:slim}.

We initialize $W$ as an $n\times n$ zero-matrix. Our algorithm iteratively fills the rows of $W$. Recall that each row of $W$ corresponds to an item $i\in \mathcal{I}$. Let $I_W\subseteq \mathcal{I}$ be the set of empty rows in $W$, which is initially the entire set $\mathcal{I}$. Furthermore, let $\hat{X}:= X - XW$ with entries $\hat{x}_{ui}$.

We now define the loss functions $l_{ij}\left(w\right):~\mathbb{R}\rightarrow\mathbb{R}$, where $i$ and $j$ are any two items in $\mathcal{I}$, as follows.
\begin{align*}
l_{ij}\left(w\right) := \lambda_1 w + \lambda_F w^2 + \sum_{u\in\mathcal{U}_{\text{train}}} \left(\hat{x}_{uj} - x_{ui}w \right)^2 
\end{align*}
Note that function $l_{ij}$ is very similar to the SLIM loss except that it is restricted to a single element in $W$. Lemma \ref{lem:gslim} describes how these loss functions and the SLIM loss correlate.

\begin{lemma}
	Given a SLIM matrix $W$ with empty rows $I_W\subseteq\mathcal{I}$, let $W'$ be another SLIM matrix that is equal to $W$ except for one row ${\textbf{w}'}_i^T$ with $i\in I_W$. We then have
	\begin{equation}
	\begin{aligned}
	l_{SLIM}\left(W'\right) = l_{SLIM}\left(W\right) + \sum_{j\in\mathcal{I}\setminus\lbrace i\rbrace} \left(l_{ij}\left(w'_{ij}\right) - l_{ij}\left(0\right)\right).
	\end{aligned}
	\label{eq:gslim}
	\end{equation}
	\label{lem:gslim}
\end{lemma}
\begin{proof}
	First, note that from $i\in I_W$ it follows by definition that $\textbf{w}_i^T$ is filled with zeros.
	Let $\hat{X}:= X - XW$ and $\hat{X}':= X - XW'$. Note that $\hat{X}' = \hat{X} - \textbf{x}_i {\textbf{w}'}_i^T$. We, thus, have
	\begin{align*}
	l_{SLIM}\left(W'\right) &= \lambda_1\lVert W + {\textbf{w}'}_i^T\rVert_1 + \lambda_F\lVert W + {\textbf{w}'}_i^T\rVert_F^2 + \lVert \hat{X} - \textbf{x}_i {\textbf{w}'}_i^T\lVert_F^2\\
	&= \lambda_1\lVert W\rVert_1 + \lambda_F\lVert W\rVert_F^2 + \lVert \hat{X}\lVert_F^2 + \lambda_1\lVert{\textbf{w}'}_i^T\rVert_1 + \lambda_F \lVert {\textbf{w}'}_i^T\lVert_F^2 \\&+ \sum_{j\in\mathcal{I}\setminus\lbrace i\rbrace}\sum_{u\in\mathcal{U}_{train}}\left(\hat{x}_{uj} - x_{ui}w'_{ij} \right)^2 - \hat{x}_{uj}^2\\
	&= l_{SLIM}\left(W\right) + \sum_{j\in\mathcal{I}\setminus\lbrace i\rbrace} \left(l_{ij}\left(w'_{ij}\right) - l_{ij}\left(0\right)\right).
	\end{align*}
\end{proof}
Our goal is to find the matrix $W'$ that minimizes $l_{SLIM}\left(W'\right)$. For a fixed item $i\in I_W$ it follows from Equation \ref{eq:gslim} that minimizing $l_{SLIM}\left(W'\right)$ is equal to minimizing the sum $\sum_{j\in\mathcal{I}\setminus\lbrace i\rbrace} l_{ij}\left(w'_{ij}\right)$. The summands $l_{ij}\left(w'_{ij}\right)$ are (for a fixed $i$) independent from each other and, thus, can be minimized separately. Furthermore, the function $l_{ij}\left(w\right)$ is a quadratic function in $w$ and we can therefore find its minimum by setting its derivative $l'_{ij}\left(w\right)$ to zero.
\begin{equation}
\begin{aligned}
l'_{ij}\left(w\right) &= \lambda_1 + 2\lambda_F w + \sum_{u\in\mathcal{U}_{\text{train}}} \left( -2\hat{x}_{uj}x_{ui} + 2x_{ui}^2w\right)\\
&= 2\left(\lambda_F + \sum_{u\in\mathcal{U}_{\text{train}}}x_{ui}^2\right)w + \lambda_1 - 2 \sum_{u\in\mathcal{U}_{\text{train}}} \hat{x}_{uj}x_{ui}\\
&\stackrel{!}{=} 0
\end{aligned}
\label{eq:derivative}
\end{equation}
Equation \ref{eq:derivative} is solved by
\begin{equation}
w^* = \frac{-\frac{\lambda_1}{2} + \sum_{u\in\mathcal{U}_{\text{train}}} \hat{x}_{uj}x_{ui}}{\lambda_F + \sum_{u\in\mathcal{U}_{\text{train}}}x_{ui}^2}.
\end{equation}
We set $w^*_{ij} := \max\lbrace \frac{-\frac{\lambda_1}{2} + \sum_{u\in\mathcal{U}_{\text{train}}} \hat{x}_{uj}x_{ui}}{\lambda_F + \sum_{u\in\mathcal{U}_{\text{train}}}x_{ui}^2}, 0\rbrace$ in order to satisfy the non-negativity constraint.

In each round of our algorithm we solve the problem
\begin{align*}
i^* = \arg\min_{i\in I_W} \sum_{j\in\mathcal{I}\setminus\lbrace i\rbrace} \left(l_{ij}\left(w^*_{ij}\right) - l_{ij}\left(0\right)\right),
\end{align*}
remove $i^*$ from $I_W$ and fill the $i^*$-th row of $W$ with the values $w^*_{i^*j}$. Note that this step also changes $\hat{X}$ and the loss functions $l_{ij}$. We repeat until $I_W$ is the empty set.

\section{Offline Experiments}
\label{sec:experiments}
\subsection{Datasets}
We consider two well-known datasets in our experiments Movielens-25M (ML-25) \cite{movielens} and the Netflix \cite{bennett2007netflix} dataset (see Table \ref{tab:data} for details). We do not filter these sets to ensure that we do not accidentally introduce any biases.
\begin{table}
	\centering
	\caption{Datasets}
	\label{tab:data}
	\begin{tabular}{ccccc}
		\toprule
		Name&Users&Items&Ratings\\
		\midrule
		ML-25\footnote{Publicly available dataset described in \cite{movielens}. Download link: \url{https://grouplens.org/datasets/movielens/25m/}}  & $162,541$ & $59,047$ & $25,000,095$\\
		Netflix\footnote{Publicly available dataset described in \cite{bennett2007netflix}.}  & $480,189$ & $17,770$ & $100,480,507$\\
		\bottomrule
	\end{tabular}
\end{table}
We split the datasets into training and test sets as described in Section \ref{subsec:evalmethod}.

\subsection{Baselines}
Our main baseline is $Q_{\text{Bandit}}$, see Section \ref{subsec:qbandit} for a description. $Q_{\text{Bandit}}$ requires a trained LFM as input and in \cite{towards} it is not specified how the LFM is trained. We decided to use the approach called PureSVD \cite{cremonesi2010performance} for the following reasons. First, PureSVD outperforms other training methods regarding top-N recommendations \cite{cremonesi2010performance}. Second, the approach is conceptually simple and there are reliable implementations available. Thus, we keep the risk low that we run into the problems discussed in \cite{difficulty}. We used the implementation called redsvd \cite{redsvd}.

Apart from $Q_{\text{Bandit}}$ we compare our method with the baselines $R_{\text{Gain}}$ and $Q_{\text{Greedy}}$ as described in Section \ref{sec:firstapproach}.

\subsection{Training Phase}
For Netflix it took 9 minutes per row to compute the SLIM matrix, for ML-25 it took even 19 minutes. It is therefore often not feasible to compute the complete SLIM matrix with this approach. Fortunately, for our purpose it suffices to compute the number of rows that equals the size of the questionnaire, which is 20 in our experiments. 

As our questionnaire is static, we can compute the SLIM matrix in a preprocessing step and do not have to worry too much  about the run time.
This is also the reason why optimizing the computational cost of the training phase is not the focus of this work. We are confident that there are many ways to improve the run´ times and leave this task as future work.

The main reason for the high computational effort is that we are unable to hold the matrix $\hat{X}$ in memory. We therefore compute it in each iteration from scratch. Furthermore, we consider each item as possible candidate for the next row, even items with very few ratings. This is the reason why the training for ML-25 takes much longer than for Netflix. Excluding lesser known items would greatly improve the run times. Lastly, the parallelization of our implementation is on a simple level, leaving much room for improvement.

Our baseline $Q_{\text{Bandit}}$ clearly outperforms $Q_{\text{GSLIM}}$ in this step. The training phase of PureSVD took $96$ seconds for Netflix and $25$ seconds for ML-25.

\subsection{Parameter Optimization}
Before conducting the actual experiments, we needed to set the parameters for $Q_{\text{GSLIM}}$ and $Q_{\text{Bandit}}$ meaningfully. The parameters of $Q_{\text{GSLIM}}$ are the regularization variables $\lambda_1$ and $\lambda_F$, while for $Q_{\text{Bandit}}$ we need to specify the number of features $\lambda_{\text{LFM}}$.

Regarding $Q_{\text{GSLIM}}$, we tried all combinations of powers of two between $2^0$ and $2^{19}$ for $\lambda_1$ and $\lambda_F$. For $Q_{\text{Bandit}}$, we considered all multiples of $100$ between $200$ and $800$ for $\lambda_{\text{LFM}}$. Note that one test run of $Q_{\text{Bandit}}$ took between $20-100$ times longer than one test run of $Q_{\text{GSLIM}}$ (due to its expensive evaluation phase), such that we roughly invested the same computational effort in both approaches.

We conducted our parameter optimization by splitting $X_{\text{train}}$ into a new test and training set following a similar procedure as explained in Section \ref{subsec:evalmethod}. Most importantly, we did not consider $X_{\text{test}}$ in this phase. Our results are shown in Table \ref{tab:paramopt}. Interestingly, our parameters for $Q_{\text{GSLIM}}$ are much larger than those reported in \cite{ning2011slim} for SLIM.

\begin{table}
	\centering
	\begin{tabular}{cccc}
		\toprule
		Questionnaire& ML-25 & Netflix \\
		\midrule
		$Q_{\text{Bandit}}$ & $\lambda_{\text{LFM}} = 700$ & $\lambda_{\text{LFM}} = 800$ \\
		$Q_{\text{GSLIM}}$ & $\lambda_1 = 2^{12}$, $\lambda_F = 2^{16}$ & $\lambda_1 = 2^{15}$, $\lambda_F = 2^{10}$ \\
		\bottomrule
	\end{tabular}
	\caption{This table shows the results of the parameter optimization for both considered datasets.}
	\label{tab:paramopt}
\end{table}

\subsection{Top-N Recommendation}
We conducted our experiments as explained in Section \ref{subsec:evalmethod}. Figure \ref{fig:experiments} shows the average results with respect to the NDCG@10. The NDCG of $Q_{\text{Bandit}}$ converges or even decreases after ten questions. We think that the main reason for the decreasing NDCG is that the questionnaires are not allowed to recommend items that they contain as questions. Therefore, using popular items as questions may have a negative impact on the recommendation performance if the questionnaire is unable to generalize the gained information. $Q_{\text{GSLIM}}$ does not suffer from this restriction. On the contrary, even after 20 questions there is no sign of convergence or overfitting.  

\begin{figure}[h]
	\centering
	\includegraphics*[width=0.7\columnwidth]{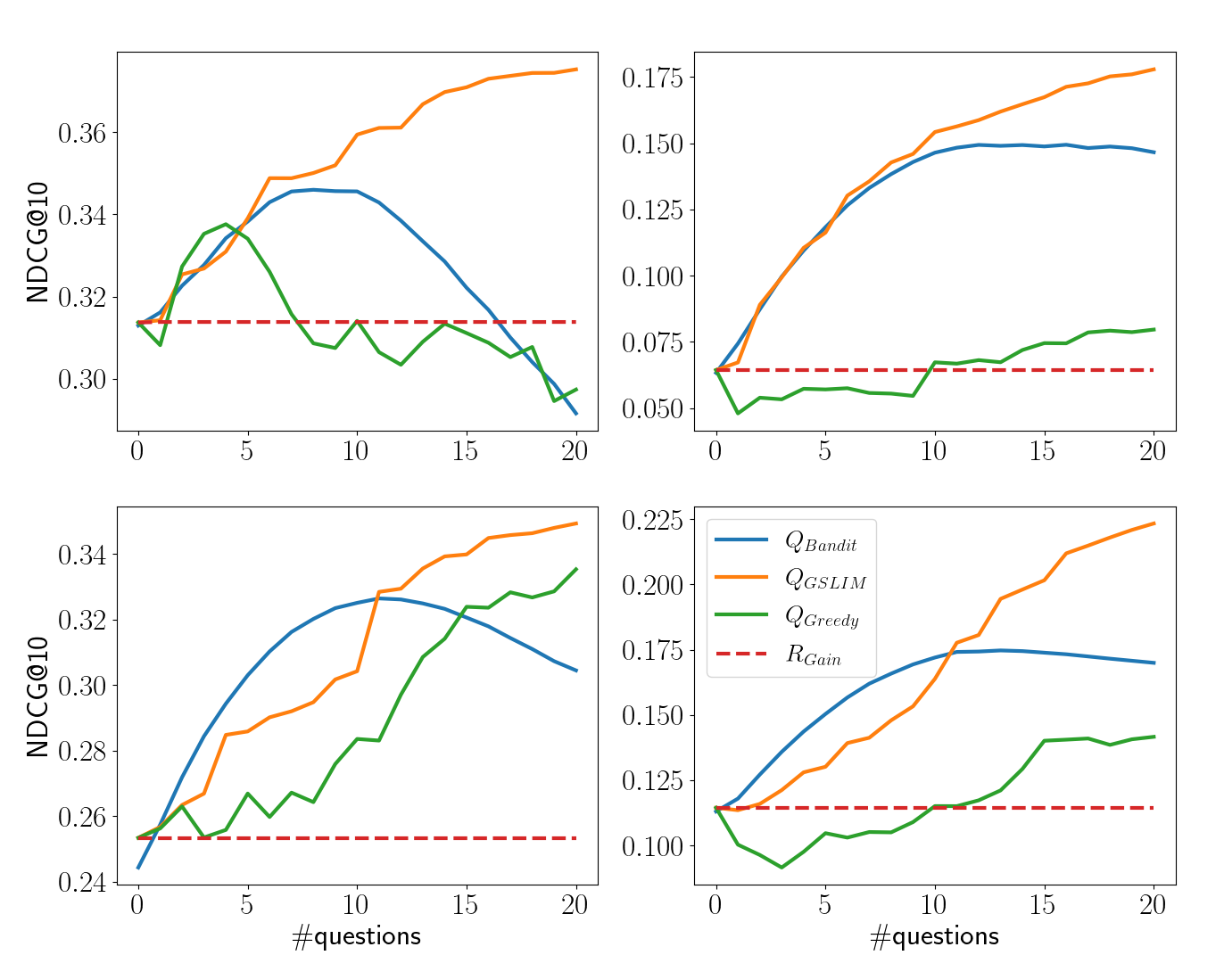}
	\caption{Comparison of our approach $Q_{\text{GSLIM}}$ with the baselines $Q_{\text{Bandit}}$, $R_{\text{Gain}}$ and $Q_{\text{Greedy}}$. Top: ML-25, bottom: Netflix, left: all items, right: long-tail items.}
	\label{fig:experiments}
\end{figure}

Looking at the numbers, shown in Table \ref{tab:expresults}, two tendencies can be observed. First, $Q_{\text{Bandit}}$ seems to perform slightly better than $Q_{\text{GSLIM}}$ if the questionnaire consists of very few questions. This is especially the case with the Netflix dataset, where $Q_{\text{Bandit}}$ clearly outperforms $Q_{\text{GSLIM}}$ during the first ten questions. Second, $Q_{\text{GSLIM}}$ is able to profit from every additional question, which leads to favorable results if the questionnaire consists of ten or more questions.

\begin{table}
	\centering
	\begin{tabular}{clcccc}
		\toprule
		&& ML-25 & & Netflix &\\
		\#Q&&$Q_{\text{GSLIM}}$&$Q_{\text{Bandit}}$&$Q_{\text{GSLIM}}$&$Q_{\text{Bandit}}$\\
		\midrule
		&NDCG@5 & $0.3480$ & $\textbf{0.3591}$ & $0.2931$ & $\textbf{0.3173}$ \\
		&NDCG@10 & $\textbf{0.3390}$ & $0.3382$ & $0.2859$ & $\textbf{0.3030}$ \\
		5&Precision@5 & $0.3959$ & $\textbf{0.4040}$ & $0.3375$ & $\textbf{0.3639}$ \\
		&Precision@10 & $\textbf{0.3686}$ & $0.3590$ & $0.3256$ & $\textbf{0.3384}$ \\
		&Recall@5 & $\textbf{0.0815}$ & $0.0805$ & $0.0367$ & $\textbf{0.0481}$ \\
		&Recall@10 & $\textbf{0.1250}$ & $0.1220$ & $0.0648$ & $\textbf{0.0783}$ \\
		\bottomrule
		&NDCG@5 & $\textbf{0.3720}$ & $0.3651$ & $0.3148$ & $\textbf{0.3425}$ \\
		&NDCG@10 & $\textbf{0.3594}$ & $0.3456$ & $0.3042$ & $\textbf{0.3251}$ \\
		10&Precision@5 & $\textbf{0.4195}$ & $0.4088$ & $0.3663$ & $\textbf{0.3907}$ \\
		&Precision@10 & $\textbf{0.3868}$ & $0.3642$ & $0.3456$ & $\textbf{0.3599}$ \\
		&Recall@5 & $\textbf{0.0880}$ & $0.0857$ & $0.0473$ & $\textbf{0.0547}$ \\
		&Recall@10 & $\textbf{0.1375}$ & $0.1314$ & $0.0777$ & $\textbf{0.0874}$ \\
		\bottomrule
		&NDCG@5 & $\textbf{0.3857}$ & $0.3388$ & $\textbf{0.3586}$ & $0.3381$ \\
		&NDCG@10 & $\textbf{0.3709}$ & $0.3221$ & $\textbf{0.3398}$ & $0.3206$ \\
		15&Precision@5 & $\textbf{0.4347}$ & $0.3785$ & $\textbf{0.4050}$ & $0.3855$ \\
		&Precision@10 & $\textbf{0.3987}$ & $0.3358$ & $\textbf{0.3756}$ & $0.3539$ \\
		&Recall@5 & $\textbf{0.0927}$ & $0.0797$ & $0.0527$ & $\textbf{0.0547}$ \\
		&Recall@10 & $\textbf{0.1437}$ & $0.1229$ & $0.0843$ & $\textbf{0.0873}$ \\
		\bottomrule
		&NDCG@5 & $\textbf{0.3895}$ & $0.3053$ & $\textbf{0.3686}$ & $0.3205$ \\
		&NDCG@10 & $\textbf{0.3752}$ & $0.2916$ & $\textbf{0.3493}$ & $0.3045$ \\
		20&Precision@5 & $\textbf{0.4399}$ & $0.3376$ & $\textbf{0.4149}$ & $0.3667$ \\
		&Precision@10 & $\textbf{0.4034}$ & $0.3004$ & $\textbf{0.3847}$ & $0.3360$ \\
		&Recall@5 & $\textbf{0.0961}$ & $0.0698$ & $\textbf{0.0548}$ & $0.0519$ \\
		&Recall@10 & $\textbf{0.1474}$ & $0.1107$ & $\textbf{0.0872}$ & $0.0830$ \\
		\bottomrule
	\end{tabular}
~\\
	\begin{tabular}{clcccc}
		\toprule
		&& ML-25 & & Netflix &\\
		\#Q&&$Q_{\text{GSLIM}}$&$Q_{\text{Bandit}}$&$Q_{\text{GSLIM}}$&$Q_{\text{Bandit}}$\\
		\midrule
		&NDCG@5 & $0.1152$ & $\textbf{0.1162}$ & $0.1359$ & $\textbf{0.1579}$ \\
		&NDCG@10 & $0.1163$ & $\textbf{0.1183}$ & $0.1301$ & $\textbf{0.1504}$ \\
		5&Precision@5 & $\textbf{0.1215}$ & $0.1177$ & $0.1572$ & $\textbf{0.1754}$ \\
		&Precision@10 & $\textbf{0.1103}$ & $0.1066$ & $0.1462$ & $\textbf{0.1599}$ \\
		&Recall@5 & $0.0291$ & $\textbf{0.0338}$ & $0.0152$ & $\textbf{0.0254}$ \\
		&Recall@10 & $0.0481$ & $\textbf{0.0558}$ & $0.0267$ & $\textbf{0.0421}$ \\
		\bottomrule
		&NDCG@5 & $\textbf{0.1539}$ & $0.1461$ & $0.1703$ & $\textbf{0.1821}$ \\
		&NDCG@10 & $\textbf{0.1543}$ & $0.1465$ & $0.1638$ & $\textbf{0.1720}$ \\
		10&Precision@5 & $\textbf{0.1570}$ & $0.1473$ & $0.1931$ & $\textbf{0.2008}$ \\
		&Precision@10 & $\textbf{0.1422}$ & $0.1312$ & $0.1788$ & $\textbf{0.1805}$ \\
		&Recall@5 & $0.0424$ & $\textbf{0.0428}$ & $0.0255$ & $\textbf{0.0316}$ \\
		&Recall@10 & $\textbf{0.0690}$ & $0.0684$ & $0.0409$ & $\textbf{0.0508}$ \\
		\bottomrule
		&NDCG@5 & $\textbf{0.1682}$ & $0.1491$ & $\textbf{0.2195}$ & $0.1848$ \\
		&NDCG@10 & $\textbf{0.1674}$ & $0.1488$ & $\textbf{0.2016}$ & $0.1739$ \\
		15&Precision@5 & $\textbf{0.1736}$ & $0.1479$ & $\textbf{0.2369}$ & $0.2029$ \\
		&Precision@10 & $\textbf{0.1555}$ & $0.1320$ & $\textbf{0.2089}$ & $0.1809$ \\
		&Recall@5 & $\textbf{0.0466}$ & $0.0437$ & $0.0323$ & $\textbf{0.0327}$ \\
		&Recall@10 & $\textbf{0.0756}$ & $0.0691$ & $0.0491$ & $\textbf{0.0519}$ \\
		\bottomrule
		&NDCG@5 & $\textbf{0.1788}$ & $0.1469$ & $\textbf{0.2444}$ & $0.1804$ \\
		&NDCG@10 & $\textbf{0.1779}$ & $0.1466$ & $\textbf{0.2234}$ & $0.1700$ \\
		20&Precision@5 & $\textbf{0.1806}$ & $0.1459$ & $\textbf{0.2620}$ & $0.1988$ \\
		&Precision@10 & $\textbf{0.1613}$ & $0.1299$ & $\textbf{0.2293}$ & $0.1767$ \\
		&Recall@5 & $\textbf{0.0499}$ & $0.0425$ & $\textbf{0.0374}$ & $0.0321$ \\
		&Recall@10 & $\textbf{0.0796}$ & $0.0680$ & $\textbf{0.0555}$ & $0.0509$ \\
		\bottomrule
	\end{tabular}
	\caption{From left to right: Number of questions, metric, results for ML-25, results for Netflix. Top: recommendations for all items, bottom: recommendations for long-tail items.}
	\label{tab:expresults}
\end{table}

In summary, while our approach $Q_{\text{GSLIM}}$ shows a strong learning curve with respect to the number of questions, it seems to depend on the application which approach to choose. If the questionnaire should be very short, $Q_{\text{Bandit}}$ is probably the preferred method. Otherwise, $Q_{\text{GSLIM}}$ performs better.

\section{User Study}
\label{sec:userstudy}

\begin{figure}[h]
	\centering
	\includegraphics*[width=.7\columnwidth]{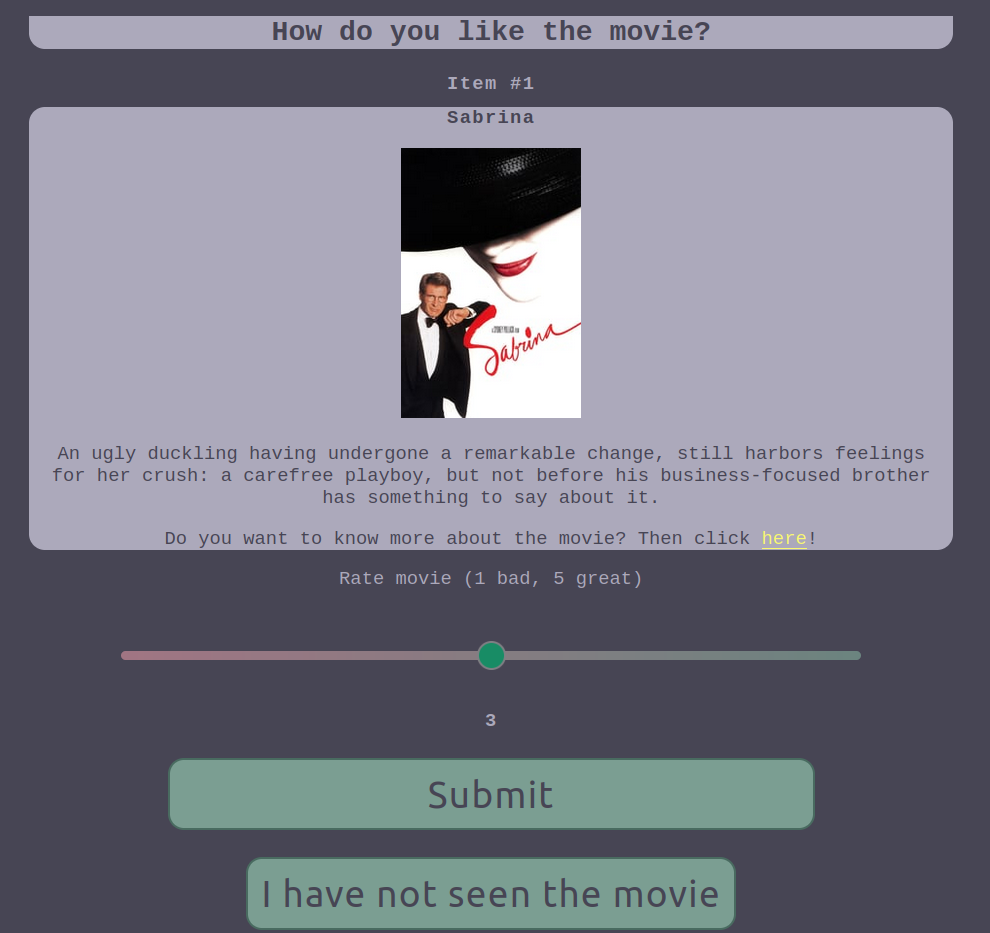}
	\caption{This figure shows the layout of a question. It contained the title, a poster and an abstract of the movie. The layout of a recommendation looked similar.}
	\label{fig:study}
\end{figure}

As discussed in Section \ref{subsec:evalmethod}, we found it necessary to conduct a user study to confirm the results of our offline experiments. Our study consisted of two parts. In the first part, the questionnaire part, the participants were asked to rate ten movies. They could give any rating from one (very bad) to five (very good) or say that they don't know the movie, which we refer to as zero-rating. 

In the second part, ten recommendations were shown and the users were asked whether they know the movie and whether they like (or think they like) the movie. There was one negative option (bad), two positive options (good, very good) and one neutral option (don't know).

Given our limited resources, we decided to compare only $Q_{\text{GSLIM}}$, $Q_{\text{Bandit}}$ and $R_{\text{Gain}}$. To do so, the participants were split randomly into two groups. The first group received the questions from $Q_{\text{GSLIM}}$ and the second group from $Q_{\text{Bandit}}$. Every participant received the top-5 recommendation from $R_{\text{Gain}}$ (which is a static recommendation) as well as the top-5 recommendation from the recommender systems that corresponds to their questionnaire. Note that the participants did neither know in which group they were, nor which algorithm computed their questions or recommendations. 

Our algorithms were trained using a subset of ML-25. We decided to exclude items with very few ratings and items produced before 1995 to make sure that our algorithms (especially $Q_{\text{Bandit}}$) are fast enough for a good user experience. We also excluded items for which we did not find the information we wanted to show the users (see Figure \ref{fig:experiments}).

Furthermore, we tried to ensure that the task of recommendation would not be too simple. First, we removed all but one representative of popular movie series. Second, all recommendations were limited to long-tail items, as explained in Section \ref{subsec:evalmethod}. The questionnaires could ask to rate short-head items, though.

The remaining dataset consisted of $162,541$ users, $8,070$ items and $13,705,543$ ratings.

Our user study had 103 participants, many of them were somehow linked to the authors. Thus, we don't claim that the study is representative. 

\subsection{Questionnaire Phase}

Interestingly, $Q_{\text{GSLIM}}$ and $Q_{\text{Bandit}}$ pursue very different strategies for the item selection. While $Q_{\text{Bandit}}$ selects items that might be good recommendations, this is not so much the case for $Q_{\text{GSLIM}}$. 

\begin{table}
	\centering
	\begin{tabular}{ccccc}
		\toprule
		Questionnaire&Known Items&Average Rating\\
		\midrule
		$Q_{\text{Bandit}}$ & $48.7$\% & $3.99$\\
		$Q_{\text{GSLIM}}$ & $34.4$\% & $3.67$\\
		\bottomrule
	\end{tabular}
	\caption{This table shows some statistics about the questionnaire phase. For instance, among all items selected by $Q_{\text{Bandit}}$ $48.7$\% received a rating other than zero. Among these items the average rating was $3.99$.}
	\label{tab:questions}
\end{table}

Table \ref{tab:questions} shows the statistics. Only $34.4$\% of the items selected by $Q_{\text{GSLIM}}$ received a rating other than zero whereas almost half of the items selected by $Q_{\text{Bandit}}$ where known to the participants. Among the known items, the items selected by $Q_{\text{Bandit}}$ also received the better ratings.

\subsection{Recommendation Phase}

Table \ref{tab:recommendations} shows the results of the recommendation phase, which are surprisingly clear. $Q_{\text{GSLIM}}$ performs remarkably well in all aspects while $Q_{\text{Bandit}}$ is not significantly better than the static recommender $R_{\text{Gain}}$.

\begin{table}
	\centering
	\begin{tabular}{lcccc}
		\toprule
		&$R_{\text{Gain}}$ & $Q_{\text{Bandit}}$ & $Q_{\text{GSLIM}}$\\
		\midrule
		Positive Feedback (PF) & $52.0$\% & $52.6$\% & $77.2$\%\\
		Very Positive Feedback & $15.7$\% & $13.9$\% & $36.1$\%\\
		Unknown Items & $54.5$\% & $51.3$\% & $25.6$\%\\
		PF Among Unknown Items & $30.0$\% & $28.0$\% & $43.8$\%\\
		\bottomrule
	\end{tabular}
	\caption{We measured for each recommender system the positive feedback, also known as hit rate, the very positive feedback, the percentage of unknown items among the recommendations and among these items the positive feedback. For instance, $52.0$\% of the recommendations generated by $R_{\text{Gain}}$ received a positive feedback by the participants, $15.7$\% received a very positive feedback and $54.5$\% of the recommended items were unknown to the participants.}
	\label{tab:recommendations}
\end{table}

The recommendations of $Q_{\text{GSLIM}}$ received a positive feedback of $77.2$\%, much more than the other two approaches. Also, if we only consider the recommendations of items unknown to the participants, the positive feedback of $43.8$\% for $Q_{\text{GSLIM}}$ is outstanding. Thus, even though the user study was small and perhaps not fully representative, there is at least a strong indication that $Q_{\text{GSLIM}}$ shows very favorable results in practice.

Regarding serendipity, one could argue that $Q_{\text{GSLIM}}$ shows worse results than the other two approaches as only $25.6$\% of the recommendations concerned unknown items. Also, regarding the product of unknown items and positive feedback among unknown items, which could be considered as the true relevant recommendations, $Q_{\text{GSLIM}}$ does not perform very well. 

From a practical point of view, this might indeed be a weak spot of $Q_{\text{GSLIM}}$ that should be improved. However, it is important to note that our design of $Q_{\text{GSLIM}}$ does not consider serendipity at all. The only objective of $Q_{\text{GSLIM}}$ (as well as of $Q_{\text{Bandit}}$) is to compute recommendations that receive positive feedback. Thus, we would blame $Q_{\text{GSLIM}}$ for doing a great job. 

Furthermore, even from a practical point of view, we do not think that the product of unknown recommendations and positive feedback among unknown items is what we want to maximize. We believe that, from a user's perspective, it is much easier to identify known items than irrelevant unknown items. Thus, the rate of positive feedback among unknown items is probably the most important measure while the rate of unknown items should not be too low.

\section{Conclusions}
\label{sec:conclusions}
In this work we discussed the idea of eliciting user preferences with SLIM, a popular approach for top-N recommendation. We first showed in Section \ref{sec:firstapproach} that the common way to train SLIM is unsuitable for this task. This motivated the presentation of a new training algorithm for SLIM called Greedy SLIM that computes the SLIM matrix greedily row by row. 

We showed in offline experiments that with Greedy SLIM it is possible to construct a questionnaire $Q_{\text{GSLIM}}$ that improves the recommendation performance for new users considerably and that it achieves better results than a popular approach based on latent factors, especially if the questionnaire contains more than ten questions.

Furthermore, we conducted a user study and were able to confirm the findings of the offline experiments also in practice. In fact, the results  of $Q_{\text{GSLIM}}$ in the user study are even more remarkable. We, hence, dare to conclude that SLIM-based questionnaires are an excellent alternative to the existing LFM-based approaches.

Even though we've tried our best to conduct a meaningful user study, we have to notice that the results might not be representative. We would therefore like to repeat this experiment on larger scale and would be grateful for support, if anyone is interested.

A small drawback of Greedy SLIM is its expensive training phase. It took hours to merely compute 20 rows of the SLIM matrix. Improving these times is a promising open problem. 

Furthermore, there are many possible extensions to $Q_{\text{GSLIM}}$. For instance, a dynamic version of $Q_{\text{GSLIM}}$ would be interesting that does not ask the same questions to each user. Moreover, it would be worthwhile investigating if the presented ideas can be applied to other types of questions such as pairwise questions.

Another interesting aspect which we did not cover is whether our greedy approach can be applied to other problem settings as well. It is possible that this training method leads to better results than the standard training approach of SLIM with coordinate descent. The big obstacle, however, is the preprocessing time of Greedy SLIM, which is even worse than with standard SLIM.
\section*{Acknowledgments}
This research was funded by the German Federal Ministry of Education and Research (BMBF) through grants 01IS17051 Software Campus 2.0.

\bibliographystyle{unsrtnat}

\end{document}